\documentclass[sn-mathphys]{sn-jnl}
\usepackage[english]{babel}
\usepackage{amsmath}
\usepackage{mathtools}
\usepackage{amssymb}
\usepackage{mathrsfs}
\usepackage{delarray}
\usepackage{tabu}
\usepackage{environ}
\usepackage{amsthm}
\usepackage{stmaryrd}
\usepackage{calc}
\usepackage{algorithm}
\usepackage{algpseudocode}
\usepackage{hyperref}

\usepackage{graphicx}
\graphicspath{ {./images/} }

\usepackage{xcolor}

\usepackage{tikz}
\usetikzlibrary{arrows, automata, positioning, arrows.meta, calc}

\tikzset{configuration/.style = {state, rectangle, minimum height=0.5cm},}

\let\emptyset\varnothing

\newcommand{\B}{\mathbb{B}}
\newcommand{\N}{\mathbb{N}}

\algblock{Input}{EndInput}
\algnotext{EndInput}
\algblock{Output}{EndOutput}
\algnotext{EndOutput}

\newtheorem{theorem}{Theorem}
\newtheorem{lemma}{Lemma}

\date{}

\begin{document}

\title{A sequential solution to the density classification task using an intermediate alphabet}


\author*[1]{\fnm{Pac\^ome} \sur{Perrotin}}\email{pacome.perrotin@gmail.com}

\author[1,2]{\fnm{Pedro} \sur{Paulo Balbi}}\email{pedrob@mackenzie.br}

\author[1,2]{\fnm{Eurico} \sur{Ruivo}}\email{eurico.ruivo@mackenzie.br}

\affil[1]{\orgdiv{P\'{o}s-Gradua\c{c}\~{a}o em Engenharia El\'{e}trica e Computa\c{c}\~{a}o}, \orgname{Universidade Presbiteriana Mackenzie}, \orgaddress{\street{Rua da Consola\c{c}\~{a}o 896}, \city{S\~{a}o Paulo}, \postcode{01302-907}, \state{SP}, \country{Brazil}}}

\affil[2]{\orgdiv{Faculdade de Computa\c{c}\~{a}o e Inform\'{a}tica}, \orgname{Universidade Presbiteriana Mackenzie}, \orgaddress{\street{Rua da Consola\c{c}\~{a}o 896}, \city{S\~{a}o Paulo}, \postcode{01302-907}, \state{SP}, \country{Brazil}}}


\abstract{
  We present a sequential cellular automaton of radius $\frac{1}{2}$
  as a solution to the density classification task that makes use of an intermediate
  alphabet, and converges to a clean fixed point with no remaining auxiliary or intermediate
  information.
  We extend this solution to arbitrary finite alphabets
  and to configurations in higher dimensions.
}


\maketitle

\section{Introduction}

Cellular automata
are computational models characterised by a local
rule uniformly applied over a configuration.
The straightforwardness of their definition makes them the subject of a wide
array of applications to the modeling of complex physical
phenomena~\cite{CA31,CA39,CA2,CA54}.

Among some of the theoretical questions
studied on cellular automata, classification
tasks~\cite{balbi11} can be compared to decision problems defined on
Turing machines.
Such a task is defined by a function that assigns any initial configuration
$x \in \B^n$
to a binary value $b$. A cellular automaton is a solution to a classification
task if while executed on the initial configuration $x$ it converges to
the uniform configuration $b^n$, and this configuration is a fixed point.
The two most studied classification tasks are the parity classification task,
where $b$ is $1$ if and only if the initial configuration contains an
odd number of $1$s, and the density classification task, where $b$
is the symbol in strict majority in the initial configuration.

Various solutions to the parity classification task exist~\cite{balbi14,Ninagawa,LeeXuChau,synchParitySolutionWith150,gen_SynchParitySolution_on_Graph},
including a pure solution in the form of a cellular automaton of
radius 4~\cite{betel2013}.
However, the density classification task has been proven to not have
such a pure solution~\cite{land1995}.
The current state of the art contains various solutions
approaching the problem from different angles,
including changing the local rule during the computation~\cite{martins2005}
or using multiple tracks
in parallel to provide a solution~\cite{kari2012}.
Solutions have also been constructed by
using a continuous space of
values~\cite{Briceno2013,Leal2023},
or by using a stochastic
approach~\cite{Fates2011}.

This paper is structured by first detailing the specifications of the
variant of the density classification task presently under investigation.
We then offer a detailed and illustrated description of our novel solution,
and lastly we offer generalisations to the solution to larger alphabets
and configurations in more dimensions.

In addition to formal proofs, we tested the one-dimensional binary
solution on all initial configurations
up to size 30, the source code for which can be freely accessed online~\cite{densityChecker},
written in Rust.

\section{The problem and the general description of the solution}

\subsection{The problem}

In this paper, we focus on a non-standard solution of the density classification task, as it relies on a variant cellular automaton and the use of an intermediate alphabet, that is, the full alphabet $\Sigma$ of the solution 
is a super set of the input alphabet, $\B$. 

The formulation of the problem is as follows: for any finite initial configuration over $\B$ that has a defined density (a symbol that outnumbers every
other one), then the lattice of the cellular automaton must converge to the uniform configuration
containing only that symbol, which must be a fixed point. If the initial solution does not have a defined density, then no expectation is set on the convergence. 

Since there is no imposition on what would happen in between the initial configuration and the fixed point, a solution to the problem is free to use any number of symbols from its extended alphabet. Solutions relying on this freedom have been found in the past, namely~\cite{kari2012}.
This
particular solution uses two tracks, one that contains the initial configuration and the second that converges to the desired fixed point, the former guiding the convergence of the latter. We consider the solution presented herein as an improvement
because it always converges to a fixed point with no remaining auxiliary or intermediate information,
whereas that two-tracks solution only ensures this clean result in one of its tracks.

\subsection{The overall idea of the solution}

Our solution is a cellular automaton of radius $\frac{1}{2}$ (the reference cell being the one on the right), updated sequentially, meaning that the local rule is applied
using the value of the current cell, plus the value of the cell to its left, while the
cellular automaton as a whole is updated from left to right. In essence, this cellular automaton simulates an operating head that moves along the configuration while updating an internal memory and transforming the configuration as it proceeds.

With its radius $\frac{1}{2}$, our solution can appear to belong
in the family of one-way cellular automata~\cite{Kutrib2014}.
However, it is updated sequentially and not in parallel,
and operates on cyclic configurations.

In the following, a cycle is a complete loop of
the operating head around the configuration.
This cellular automaton tries to remove exactly one $0$ and one $1$ from the
configuration per cycle, storing them in its internal memory.
It then discards this pair when a cycle is
complete. The detection of the end of a cycle is done thanks to an internal counter.
At the end of a cycle, if only one symbol was found (say, a $0$), it then means that there
is no more of the other symbol in the configuration (so, no more $1$s). Since symbols
have been removed in equal amounts, the cellular automaton can now converge to the uniform
configuration containing the symbol it
still holds in memory (here, it would converge to $0^n$).
All these extra pieces of information -- the internal memory, the internal counter, and the marker
that allows to know which symbols have been removed from the configuration -- are encoded using the
aforementioned intermediate alphabet.

So, alongside the Boolean alphabet $\B = \{0, 1\}$, based upon which the initial configuration and the final fixed points are defined, a key aspect of  the solution is that it relies on an intermediate alphabet composed of the following
16 triplets:

\vspace{0.1cm}
\noindent
$\begin{psmallmatrix}\circ \\ 0 \\ \{0\}\end{psmallmatrix}$,
$\begin{psmallmatrix}\circ \\ 0 \\ \B\end{psmallmatrix}$,
$\begin{psmallmatrix}\circ \\ 1 \\ \{1\}\end{psmallmatrix}$,
$\begin{psmallmatrix}\circ \\ 1 \\ \B\end{psmallmatrix}$,
$\begin{psmallmatrix}\circ \\ X \\ 
\emptyset\end{psmallmatrix}$,
$\begin{psmallmatrix}\circ \\ X \\ \{0\}\end{psmallmatrix}$,
$\begin{psmallmatrix}\circ \\ X \\ \{1\}\end{psmallmatrix}$,
$\begin{psmallmatrix}\circ \\ X \\ \B\end{psmallmatrix}$,
$\begin{psmallmatrix}\bullet \\ 0 \\ \{0\}\end{psmallmatrix}$,
$\begin{psmallmatrix}\bullet \\ 0 \\ \B\end{psmallmatrix}$,
$\begin{psmallmatrix}\bullet \\ 1 \\ \{1\}\end{psmallmatrix}$,

\ \\ 
\noindent
$\begin{psmallmatrix}\bullet \\ 1 \\ \B\end{psmallmatrix}$,
$\begin{psmallmatrix}\bullet \\ X \\ \emptyset\end{psmallmatrix}$,
$\begin{psmallmatrix}\bullet \\ X \\ \{0\}\end{psmallmatrix}$,
$\begin{psmallmatrix}\bullet \\ X \\ \{1\}\end{psmallmatrix}$
and
$\begin{psmallmatrix}\bullet \\ X \\ \B\end{psmallmatrix}$.

\vspace{3pt}
Therefore, the intermediate alphabet is actually composed of triplets of symbols, organised in three layers, one on top of the other. 

At the top, the symbols $\{\circ, \bullet\}$ form an internal counter of the number of cycles.
Each cycle has a uniform top layer, with either only $\circ$ symbols or only $\bullet$ symbols,
and is always succeeded by a cycle with a uniform top layer using the opposite symbol. 
Alternating between these two symbols at each new cycle allows the solution to tell if
any triplet and its left neighbour originate from the same cycle.

In the middle layer, the symbol in $\{0, 1, X\}$
represents the value contained at any particular location of the configuration. The symbol $X$ indicates the positions in the configuration where a pair of $0$ and $1$ have been removed.

As for the bottom layer, the symbols $\{\emptyset, \{0\}, \{1\}, \B\}$ are related to the memory of the moving head of the machine encoded in the cellular automaton. This memory is a set that can store one individual $0$ and one individual $1$, which are taken from the configuration.
A memory of value $\emptyset$ means no symbol has yet been taken from the configuration in
the current cycle, while a memory of $\{0\}$ (respectively, $\{1\}$) means that only a $0$ symbol
(respectively, a $1$ symbol) has been taken in the current cycle,
and a memory of $\B$ means both a $0$ and a $1$ have been taken in the current cycle.
While a memory symbol appears at the bottom layer of every intermediate symbol of a configuration,
not all of them are meaningful. The memory symbol considered to be the actual memory
of the automaton (its \emph{active memory}) is the one that has been added most recently in the configuration.

Based on the previous description of the layers, it is clear that there are indeed $2 \times 3 \times 4 = 24$ possible triplets; however, since the triplets 
$\begin{psmallmatrix}\circ \\ 0 \\ \emptyset\end{psmallmatrix}$,
$\begin{psmallmatrix}\circ \\ 0 \\ \{1\}\end{psmallmatrix}$,
$\begin{psmallmatrix}\circ \\ 1 \\ \emptyset\end{psmallmatrix}$,
$\begin{psmallmatrix}\circ \\ 1 \\ \{0\}\end{psmallmatrix}$,
$\begin{psmallmatrix}\bullet \\ 0 \\ \emptyset\end{psmallmatrix}$,
$\begin{psmallmatrix}\bullet \\ 0 \\ \{1\}\end{psmallmatrix}$,
$\begin{psmallmatrix}\bullet \\ 1 \\ \emptyset\end{psmallmatrix}$ and
$\begin{psmallmatrix}\bullet \\ 1 \\ \{0\}\end{psmallmatrix}$
are never used,
they are not considered part of the intermediate alphabet, bringing the total down from 24 to 16. 

A warning should be made to the reader at this point: from now on we will refer to the triplets of the intermediate alphabet as its \emph{symbols}, which should not be confused with the individual symbols they are made of, as already mentioned; so, in a sense we will be stretching the meaning of the word `symbol' associated to the components of the intermediate alphabet, just to preserve the standard jargon associated to the notion of an alphabet, that is usually said to be composed of symbols.

\subsection{Mechanics of the process}

Since the constructed cellular automaton must have $0^n$ and $1^n$ as fixed points, this implies that the state transitions 
$00 \mapsto 0$ and $11 \mapsto 1$ must be part of the local rule.
This represents an action embedded in the rule that we can refer to as \emph{fixed point preservation}; other actions are represented in phases of the process which we refer to as \emph{kickstart}, \emph{propagation}, \emph{swap} and \emph{convergence}, as explained below.

\begin{itemize}

\item The \emph{kickstart} phase is applied when $01$ or $10$ is detected on the initial configuration.
These patterns are proof that the configuration is not uniform, and thus the cellular automaton
kickstarts its computation of the density of the configuration by inserting a symbol
from the intermediate alphabet. In particular, the neighbourhood $01$ maps to $\begin{psmallmatrix} \circ \\ X \\ \{1\} \end{psmallmatrix}$
and $10$ maps to $\begin{psmallmatrix} \circ \\ X \\ \{0\} \end{psmallmatrix}$.
In the $01$ case, the resulting symbol $\begin{psmallmatrix} \circ \\ X \\ \{1\} \end{psmallmatrix}$
shows that the symbol $1$ is no longer in the
configuration at this location, has been replaced
by $X$, and has been placed in the active memory $\{1\}$.
Furthermore, we initialise the internal counter
to the value $\circ$, which indicates
an odd number of cycles.
An equivalent
logic applies to the case $10$, where the $0$ is removed and placed in memory instead.\\

\item The \emph{propagation} phase applies directly after the kickstart phase, but also in any case
of the form
$ab$ where $a$ is an intermediate symbol and $b$ is either in $\B$ or is an
intermediate symbol, but with a different internal counter from $a$.
The propagation phase represents the majority of the updates once the process has started,
and its function is to propagate the counter from the left to the right, and also to
remove the current character of the configuration and to add it to the memory, if possible.\\

For example, the pair $\begin{psmallmatrix} \circ \\ X \\ \{1\}\end{psmallmatrix}\ 0$ maps to the symbol $\begin{psmallmatrix} \circ \\ X \\ \B \end{psmallmatrix}$.
In essence, the current symbol $0$ has been replaced by an intermediate symbol,
in which the symbol $0$ has been removed and placed in the memory, and
the counter $\circ$ has been propagated from the left. The memory is now $\B = \{0, 1\}$,
which means a complete pair has been collected.
Another typical case of the propagation phase is associated to the neighbourhood $\begin{psmallmatrix} \bullet \\ 0 \\ \{0\} \end{psmallmatrix}\ \begin{psmallmatrix} \circ \\ 0 \\ \B \end{psmallmatrix}$,
which maps to $\begin{psmallmatrix} \bullet \\ 0 \\ \{0\} \end{psmallmatrix}$; here, the new character is obtained by propagating
the counter and the memory from the left. As the current character $0$ is already present
in the memory, it is not removed. Note that the previous memory
$\B$ date from the previous cycle, as indicated by the outdated counter $\circ$, and can
thus be deleted in the process of the propagation.\\

To reiterate, the propagation phase serves two purposes: to propagate the current counter
and the active memory
from left to right, and to remove any character it can out of the configuration and place
them in the active memory.\\

\item The \emph{swap} phase happens any time where the updated cell is at the right hand side of another cell with an
identical counter, such as in $\begin{psmallmatrix} \circ \\ 0 \\ \B \end{psmallmatrix}\ \begin{psmallmatrix} \circ \\ X\\ \{0\}\end{psmallmatrix}$.
What happens now depends on the state of the active memory. As a reminder, the active
memory is the memory contained in the cell to the left of the cell which is to be updated. Since in this example, the active memory is $\B$, it is erased and the internal counter increased, therefore leading to the symbol $\begin{psmallmatrix} \bullet \\ X \\ \emptyset \end{psmallmatrix}$. This serves two purposes:
the \{$0, 1$\} pair that was harvested from the configuration in the previous cycle is deleted,
and the internal counter is updated, allowing for another propagation phase.\\

Two remaining cases are possible:
the active memory could be empty, or a singleton. If the memory is empty, this is the situation 
where the initial configuration had an equal amount of $0$s and $1$s. As discussed previously,
the solution is not required to converge in this case and thus this is classified as an undefined behaviour.
However, if the active memory is a singleton $\{b\}$,
the automaton now has to converge to the uniform configuration $b^n$. This is done by
simply updating to the symbol $b$. For example, the neighbourhood
$\begin{psmallmatrix} \bullet\\ 1\\ \{1\} \end{psmallmatrix}\ \begin{psmallmatrix} \bullet\\ X\\ \emptyset\end{psmallmatrix}$ would be updated to the symbol $1$. The rest
of the convergence is achieved through the next and final phase.\\

\item The \emph{convergence} phase is straightforward, and happens at any time where the left hand symbol
is in $\B$ and the one at the right is an intermediate symbol, in which case 
the left hand symbol is simply propagated, as happens, for example, with the neighbourhood $0\ \begin{psmallmatrix} \circ \\ X \\ \{0\} \end{psmallmatrix}$ which maps to $0$.
This simple propagation method leads to a uniform configuration, in our example $0^n$,
which is a fixed point and concludes the execution of the process.

\end{itemize}

A more formal specification of the local rule is presented in Table~\ref{figure-table}. Executions of the solution on configurations
of size $7$ and $13$ are represented in Figures~\ref{figure-exec7}
and~\ref{figure-exec13}, respectively.

\begin{table}[t!]
  \begin{center}
    {
\tabulinesep=1.2mm
\begin{tabular}{ | c | c | c | } 
  \hline
  \bf{Pattern} & \bf{Case name} & \bf{Example} \\
  \hline
  $00 \mapsto 0$, $11 \mapsto 1$ & Fixed point behaviour & $00 \mapsto 0$ \\
  \hline
  0 1 $\mapsto \begin{psmallmatrix} \circ \\ X \\ \{1\} \end{psmallmatrix}$,
  1 0 $\mapsto \begin{psmallmatrix} \circ \\ X \\ \{0\} \end{psmallmatrix}$
  & Kickstart &
  0 1 $\mapsto \begin{psmallmatrix} \circ \\ X \\ \{1\} \end{psmallmatrix}$ \\
  \hline
  $\begin{smallmatrix}\forall M \subseteq \B,\\b \in \B \setminus M,\\c \in \{\bullet, \circ\},\end{smallmatrix}$
  $\begin{psmallmatrix} c \\ ? \\ M \end{psmallmatrix}\ b
  \mapsto
  \begin{psmallmatrix} c \\ X \\ M \cup \{b\} \end{psmallmatrix}$
  & \begin{tabular}{@{}c@{}} First propagation phase\\and removing a character\end{tabular} &
  $\begin{psmallmatrix} \circ \\ 0 \\ \{0\} \end{psmallmatrix}\ 1
  \mapsto
  \begin{psmallmatrix} \circ \\ X \\ \B \end{psmallmatrix}$
  \\
  \hline
  $\begin{smallmatrix}\forall M \subseteq \B,\\b \in M,\\c \in \{\bullet, \circ\},\end{smallmatrix}$
  $\begin{psmallmatrix} c \\ ? \\ M \end{psmallmatrix}\ b
  \mapsto
  \begin{psmallmatrix} c \\ b \\ M \end{psmallmatrix}$
  & \begin{tabular}{@{}c@{}} First propagation phase,\\ without removing\\a character\end{tabular} &
  $\begin{psmallmatrix} \bullet \\ 1 \\ \{1\} \end{psmallmatrix}\ 1
  \mapsto
  \begin{psmallmatrix} \bullet \\ 1 \\ \{1\} \end{psmallmatrix}$
  \\
  \hline
  $\begin{smallmatrix}\forall c, c' \in \{\bullet, \circ\},\\c \neq c',\end{smallmatrix}$
  $\begin{psmallmatrix} c \\ ? \\ \B \end{psmallmatrix}
  \begin{psmallmatrix} c \\ ? \\ ? \end{psmallmatrix}
  \mapsto
  \begin{psmallmatrix} c' \\ X \\ \emptyset \end{psmallmatrix}$
  & \begin{tabular}{c} Swap phase and\\ resetting the memory \end{tabular} &
  $\begin{psmallmatrix} \bullet \\ X \\ \B \end{psmallmatrix}
  \begin{psmallmatrix} \bullet \\ X \\ \{0\} \end{psmallmatrix}
  \mapsto
  \begin{psmallmatrix} \circ \\ X \\ \emptyset \end{psmallmatrix}$
  \\
  \hline
  $\begin{smallmatrix}\forall M \subseteq \B,\\b \in \B \setminus M,\\c, c' \in \{\bullet, \circ\},\\c \neq c',\end{smallmatrix}$
  $\begin{psmallmatrix} c \\ ? \\ M \end{psmallmatrix}
  \begin{psmallmatrix} c' \\ b \\ ? \end{psmallmatrix}
  \mapsto
  \begin{psmallmatrix} c \\ X \\ M \cup \{b\} \end{psmallmatrix}$
  & \begin{tabular}{c} Propagation phase and\\ removing a character\end{tabular} &
  $\begin{psmallmatrix} \bullet \\ X \\ \emptyset \end{psmallmatrix}
  \begin{psmallmatrix} \circ \\ 1 \\ \B \end{psmallmatrix}
  \mapsto
  \begin{psmallmatrix} \bullet \\ X \\ \{1\} \end{psmallmatrix}$
  \\
  \hline
  $\begin{smallmatrix}\forall M \subseteq \B,\\b \in \B \setminus M,\\c, c' \in \{\bullet, \circ\},\\c \neq c',\end{smallmatrix}$
  $\begin{psmallmatrix} c \\ ? \\ M \end{psmallmatrix}
  \begin{psmallmatrix} c' \\ b \\ ? \end{psmallmatrix}
  \mapsto
  \begin{psmallmatrix} c \\ b \\ M \end{psmallmatrix}$
  & \begin{tabular}{@{}c@{}} Propagation phase,\\without removing\\a character\end{tabular} &
  $\begin{psmallmatrix} \circ \\ X \\ \B \end{psmallmatrix}
  \begin{psmallmatrix} \bullet \\ X \\ \{0\} \end{psmallmatrix}
  \mapsto
  \begin{psmallmatrix} \circ \\ X \\ \B \end{psmallmatrix}$
  \\
  \hline
  $\begin{smallmatrix}\forall b \in \B,\\c \in \{\bullet, \circ\},\end{smallmatrix}$
  $\begin{psmallmatrix} c \\ ? \\ \{b\} \end{psmallmatrix}
  \begin{psmallmatrix} c \\ ? \\ ? \end{psmallmatrix}
  \mapsto b$
  & \begin{tabular}{c} Swap phase and\\ starting the convergence\end{tabular} &
  $\begin{psmallmatrix} \circ \\ 0 \\ \{0\} \end{psmallmatrix}
  \begin{psmallmatrix} \circ \\ X \\ \emptyset \end{psmallmatrix}
  \mapsto 0$
  \\
  \hline
  $\forall b \in \B,$
  $b\ 
  \begin{psmallmatrix} ? \\ ? \\ ? \end{psmallmatrix}
  \mapsto b$
  & Convergence phase &
  $1\ 
  \begin{psmallmatrix} \circ \\ 1 \\ \B \end{psmallmatrix}
  \mapsto 1$
  \\
  \hline
\end{tabular}
}
  \end{center}
  \caption{
    \label{figure-table}
    Table detailing the different cases of the proposed solution to the density classification task
    using an intermediate alphabet. The propagation phase is described by $4$ different
    sets of cases, and the swap phase by $2$.
    The symbol \lq$?$\rq\ in patterns represents a wildcard.
  }
\end{table}

\begin{figure}[t!]
  \begin{center}
    \begin{tikzpicture}
  \newcommand{\xgap}{0.5}
  \newcommand{\ygap}{1.3}

  \draw[gray!50!white] foreach \x in {0, ..., 7} {
    (\x * \xgap - \xgap / 2, - \ygap / 2) -- (\x * \xgap - \xgap / 2, \ygap * 5.5)
  };
  \draw[gray!50!white] foreach \y in {0, ..., 6} {
    (- \xgap / 2, \y * \ygap - \ygap / 2) -- (\xgap * 6.5, \y * \ygap - \ygap / 2)
  };

  \foreach \x\y\nodetext
  in {
    0/0/0, 1/0/0, 2/0/0, 3/0/0, 4/0/0, 5/0/0, 6/0/0,
    3/1/0, 4/1/0, 5/1/0, 6/1/0,
    0/4/0, 1/4/0, 2/4/0,
    0/5/0, 1/5/0, 2/5/0, 3/5/1, 4/5/0, 5/5/1, 6/5/0
  }
  {
    \draw (\x * \xgap, \y * \ygap) node {\nodetext};
  }

  \foreach \x\y\tapesymbol\memory
  in {
    0/2/0/\B, 1/2/0/\B, 2/2/0/\B,
    3/3/X/\emptyset, 4/3/X/\emptyset, 5/3/X/\{1\}, 6/3/X/\B
  }
  {
    \draw (\x * \xgap, \y * \ygap) node {$\begin{matrix} \bullet \\ \tapesymbol \\ \memory \end{matrix}$};
  }

  \foreach \x\y\tapesymbol\memory
  in {
    0/1/X/\{0\}, 1/1/0/\{0\}, 2/1/0/\{0\},
    3/2/X/\emptyset, 4/2/X/\emptyset, 5/2/X/\emptyset, 6/2/X/\emptyset,
    0/3/0/\B, 1/3/0/\B, 2/3/0/\B,
    3/4/X/\{1\}, 4/4/X/\B, 5/4/1/\B, 6/4/0/\B
  }
  {
    \draw (\x * \xgap, \y * \ygap) node {$\begin{matrix} \circ \\ \tapesymbol \\ \memory \end{matrix}$};
  }

\end{tikzpicture}
  \end{center}
  \caption{
    \label{figure-exec7}
    An execution of the proposed solution to a configuration of size $7$, in the
    form of a table. The first row contains the initial configuration,
    and the execution operates from top to bottom.
    The first propagation phase starts at line $2$, column $4$. This
    propagation phase is followed by a second one with internal counter
    $\bullet$, and together they remove two symbols $0$ and two symbols
    $1$ from the configuration, leaving only $0$ symbols. The third
    propagation phase ends with the active memory $\{0\}$, which
    leads to the convergence phase starting at line $5$.
  }
\end{figure}

\begin{figure}[t!]
  \begin{center}
    \scalebox{0.9}{
\begin{tikzpicture}
  \newcommand{\xgap}{0.5}
  \newcommand{\ygap}{-1.3}

  \draw[gray!50!white] foreach \x in {0, ..., 13} {
    (\x * \xgap - \xgap / 2, - \ygap / 2) -- (\x * \xgap - \xgap / 2, \ygap * 9.5)
  };
  \draw[gray!50!white] foreach \y in {0, ..., 10} {
    (- \xgap / 2, \y * \ygap - \ygap / 2) -- (\xgap * 12.5, \y * \ygap - \ygap / 2)
  };

  \foreach \x\y\nodetext
  in {
    0/0/0, 1/0/0, 2/0/0, 3/0/1, 4/0/0, 5/0/1, 6/0/1, 7/0/0, 8/0/1, 9/0/1, 10/0/0, 11/0/1, 12/0/0,
    0/1/0, 1/1/0, 2/1/0,
    3/8/0, 4/8/0, 5/8/0, 6/8/0, 7/8/0, 8/8/0, 9/8/0, 10/8/0, 11/8/0, 12/8/0,
    0/9/0, 1/9/0, 2/9/0, 3/9/0, 4/9/0, 5/9/0, 6/9/0, 7/9/0, 8/9/0, 9/9/0, 10/9/0, 11/9/0, 12/9/0
  }
  {
    \draw (\x * \xgap, \y * \ygap) node {\nodetext};
  }

  \foreach \x\tapesymbol\memory
  in {
    3/X/\{1\}, 4/X/\B, 5/1/\B, 6/1/\B, 7/0/\B,
    8/1/\B, 9/1/\B, 10/0/\B, 11/1/\B, 12/0/\B
  }
  {
    \draw (\x * \xgap, 1 * \ygap)
      node {$\begin{matrix}
        \circ \\
        \tapesymbol \\
        \memory
      \end{matrix}$};
  }

  \foreach \x\tapesymbol\memory
  in {
    0/0/\B, 1/0/\B, 2/0/\B
  }
  {
    \draw (\x * \xgap, 2 * \ygap)
      node {$\begin{matrix}
        \circ \\
        \tapesymbol \\
        \memory
      \end{matrix}$};
  }

  \foreach \x\tapesymbol\memory
  in {
    3/X/\emptyset, 4/X/\emptyset, 5/X/\{1\}, 6/1/\{1\}, 7/X/\B,
    8/1/\B, 9/1/\B, 10/0/\B, 11/1/\B, 12/0/\B
  }
  {
    \draw (\x * \xgap, 2 * \ygap)
      node {$\begin{matrix}
        \bullet \\
        \tapesymbol \\
        \memory
      \end{matrix}$};
  }

  \foreach \x\tapesymbol\memory
  in {
    0/0/\B, 1/0/\B, 2/0/\B
  }
  {
    \draw (\x * \xgap, 3 * \ygap)
      node {$\begin{matrix}
        \bullet \\
        \tapesymbol \\
        \memory
      \end{matrix}$};
  }

  \foreach \x\tapesymbol\memory
  in {
    3/X/\emptyset, 4/X/\emptyset, 5/X/\emptyset, 6/X/\{1\}, 7/X/\{1\},
    8/1/\{1\}, 9/1/\{1\}, 10/X/\B, 11/1/\B, 12/0/\B
  }
  {
    \draw (\x * \xgap, 3 * \ygap)
      node {$\begin{matrix}
        \circ \\
        \tapesymbol \\
        \memory
      \end{matrix}$};
  }

  \foreach \x\tapesymbol\memory
  in {
    0/0/\B, 1/0/\B, 2/0/\B
  }
  {
    \draw (\x * \xgap, 4 * \ygap)
      node {$\begin{matrix}
        \circ \\
        \tapesymbol \\
        \memory
      \end{matrix}$};
  }

  \foreach \x\tapesymbol\memory
  in {
    3/X/\emptyset, 4/X/\emptyset, 5/X/\emptyset, 6/X/\emptyset, 7/X/\emptyset,
    8/X/\{1\}, 9/1/\{1\}, 10/X/\{1\}, 11/1/\{1\}, 12/X/\B
  }
  {
    \draw (\x * \xgap, 4 * \ygap)
      node {$\begin{matrix}
        \bullet \\
        \tapesymbol \\
        \memory
      \end{matrix}$};
  }

  \foreach \x\tapesymbol\memory
  in {
    0/0/\B, 1/0/\B, 2/0/\B
  }
  {
    \draw (\x * \xgap, 5 * \ygap)
      node {$\begin{matrix}
        \bullet \\
        \tapesymbol \\
        \memory
      \end{matrix}$};
  }

  \foreach \x\tapesymbol\memory
  in {
    3/X/\emptyset, 4/X/\emptyset, 5/X/\emptyset, 6/X/\emptyset, 7/X/\emptyset,
    8/X/\emptyset, 9/X/\{1\}, 10/X/\{1\}, 11/1/\{1\}, 12/X/\{1\}
  }
  {
    \draw (\x * \xgap, 5 * \ygap)
      node {$\begin{matrix}
        \circ \\
        \tapesymbol \\
        \memory
      \end{matrix}$};
  }

  \foreach \x\tapesymbol\memory
  in {
    0/X/\B, 1/0/\B, 2/0/\B
  }
  {
    \draw (\x * \xgap, 6 * \ygap)
      node {$\begin{matrix}
        \circ \\
        \tapesymbol \\
        \memory
      \end{matrix}$};
  }

  \foreach \x\tapesymbol\memory
  in {
    3/X/\emptyset, 4/X/\emptyset, 5/X/\emptyset, 6/X/\emptyset, 7/X/\emptyset,
    8/X/\emptyset, 9/X/\emptyset, 10/X/\emptyset, 11/X/\{1\}, 12/X/\{1\}
  }
  {
    \draw (\x * \xgap, 6 * \ygap)
      node {$\begin{matrix}
        \bullet \\
        \tapesymbol \\
        \memory
      \end{matrix}$};
  }

  \foreach \x\tapesymbol\memory
  in {
    0/X/\{1\}, 1/X/\B, 2/0/\B
  }
  {
    \draw (\x * \xgap, 7 * \ygap)
      node {$\begin{matrix}
        \bullet \\
        \tapesymbol \\
        \memory
      \end{matrix}$};
  }

  \foreach \x\tapesymbol\memory
  in {
    3/X/\emptyset, 4/X/\emptyset, 5/X/\emptyset, 6/X/\emptyset, 7/X/\emptyset,
    8/X/\emptyset, 9/X/\emptyset, 10/X/\emptyset, 11/X/\emptyset, 12/X/\emptyset
  }
  {
    \draw (\x * \xgap, 7 * \ygap)
      node {$\begin{matrix}
        \circ \\
        \tapesymbol \\
        \memory
      \end{matrix}$};
  }

  \foreach \x\tapesymbol\memory
  in {
    0/X/\emptyset, 1/X/\emptyset, 2/X/\{0\}
  }
  {
    \draw (\x * \xgap, 8 * \ygap)
      node {$\begin{matrix}
        \circ \\
        \tapesymbol \\
        \memory
      \end{matrix}$};
  }

\end{tikzpicture}
}
  \end{center}
  \caption{
    \label{figure-exec13}
    An execution of the proposed solution to a configuration of size $13$, in the
    form of a table. The first row contains the initial configuration,
    and the execution operates from top to bottom.
    This execution contains a total of $7$ propagation phases,
    which is also the count of $1$ symbols in the starting configuration,
    plus one.
  }
\end{figure}

\section{Proofs}

For what follows, we first introduce a few technical terms:
\begin{itemize}
\item $F$ is the solution previously introduced
in the paper. 
\item For any configuration $x$ of length $n$, $F(x)$ denotes the update of $x$
through the solution. 
\item For any integer $k \in \N$, $F^k(x)$ is the repetition of
the application of $F$ on $x$,
$k$ times. 
\item Since our solution works in the sequential update schedule, we will use a notation
to denote parts of an update. For any integer $i \in \{1, \ldots, n\}$, $F^{k, i}(x)$ denotes
the configuration obtained after $k$ applications of function $F$, followed by the application
of $F$'s local rule on cells $1$ to $i$, in order. In effect, $F^{k, n}(x)$ is equivalent
to $F^{k + 1}(x)$.
We denote $|x|_0$ and $|x|_1$ the number of $0$ and $1$ symbols
in the configuration $x$.
\end{itemize}

We call the \emph{active cell} the cell to the left of the next cell to be updated.
If we consider some configuration $F^{k, i}(x)$ taken at some point of an execution,
the active cell is the cell number $i$. If the symbol contained in that cell is
from the intermediate
alphabet, its memory part is called the \emph{active memory}
of the automata.

For the lemma that follows, we specify that a propagation phase starts when the first symbol
with an updated internal counter is inserted,
and lasts until the configuration has a uniform internal counter.
The convergence phase at the end of an execution is not considered as a propagation phase.
We refer to as a rule case as an application of the local rule. For example,
by propagation rule case we mean the application of the local rule in the propagation
phase.

\begin{lemma}
  \label{lemma-invariant}
  Throughout any propagation phase, the count of $0$ and $1$ symbols
  in the configuration
  plus the symbols contained in the active memory are preserved.
\end{lemma}

\begin{proof}
  At the start of the propagation phase, the active memory is $\emptyset$. Thus, the count
  of $0$ and $1$ symbols at this point is exactly what is contained in the configuration.
  Starting from there and at each step, the local rule attempts to remove the local symbol
  in the configuration and to place it in memory, if possible. If this is the case, the removed
  symbol would then be contained in the active memory, preserving the count.
  If no symbol is taken, the active memory is simply copied over and the configuration is
  not altered, thus conserving the count. 
\end{proof}

\begin{lemma}
  \label{lemma-count}
  Let $x$ be a configuration and $k$ and $i$ be integers such that
  $F^{k, i}(x)$ is a configuration which has completed $p$ propagation
  phases and has not started a convergence phase.
  The count of $0$ and $1$ symbols in the configuration
  and the active memory of $F^{k, i}(x)$ are $|x|_0 - p$ and $|x|_1 - p$.
\end{lemma}

\begin{proof}
We observe that at the end of a propagation phase, if the active memory is $\B$, 
a swap is applied that starts a new propagation phase and
  resets the active memory 
to $\emptyset$. Thus, if at the end of a propagation phase the count of the $0$ and $1$ symbols
  shared between the configuration and the active memory was
  $(|F^{k, i}(x)|_0, |F^{k, i}(x)|_1)$,
the count starting in the new propagation phase has to be
  $(|F^{k, i}(x)|_0 - 1, |F^{k, i}(x)|_1) - 1$.
  Lemma~\ref{lemma-invariant} then tells us that this count is preserved until the start
  of the next propagation phase.

  Assuming that $|x|_0$ and $|x|_1$ are both higher than $p$,
  a sequence of $p$ propagation phases is observed,
  which implies a final count of $|x|_0 - p$ and $|x|_1 - p$.
\end{proof}

\begin{theorem}
  The sequential
  cellular automaton $F$ is a solution to the density classification task using an intermediate
  alphabet.
\end{theorem}

\begin{proof}
  Let $x \in \B^n$ be an initial configuration.
  If $|x|_0 = |x|_1$, then no special behaviour is expected. We now assume that $|x_0| \neq |x_1|$.
  If $x$ is uniform in value, then it is a fixed point of $F$, which is a required behaviour.
  We now assume that $x$ is not uniform.

  This implies that a pattern $01$ or $10$ appears somewhere in the configuration,
  which implies that the kickstart rule case applies at some point, which in turns implies
  a number of propagation phases. We now argue that exactly $min(|x|_0, |x|_1) + 1$ propagation phases
  happen in the execution, followed by a convergence phase.

  As shown in Lemma~\ref{lemma-count}, the succession of $p$ propagation phases removes
  $p$ symbols $0$ and $1$ from the configuration and the active memory.
  We argue that as long as the configuration and active memory contains a symbol $0$ and $1$
  at the beginning of a propagation phase, then the active memory will be $\B$ at the end
  of a propagation phase. This can be seen because all propagation rule cases
  try to capture symbols from the configuration into the memory at any occasions, and all
  symbols of the configurations are considered this way during a propagation phase.
  Thus, as long as the configuration and active memory contains both a $0$ and $1$,
  another propagation phase will follow the current one. It follows that
  a sequence of $min(|x|_0, |x|_1)$ propagation phases happen in a sequence, after which there will be
  no more $0$ or $1$ left in the configuration and active memory, which leads to an extra
  propagation phase which ends with an active memory of the form $\{b\}$, 
  which then leads
  to a convergence phase.

  When this is the case, we argue that the remaining symbol (either $0$ and $1$) had a greater
  count in the initial configuration than the other one, which implies that it was the 
  majority symbol in the initial configuration.
  Furthermore, at the end of the last configuration phase, the active memory $\{b\}$ contains
  the symbol in majority, and the following convergence phase turns the configuration
  into the fixed point $b^n$, which proves that $F$ is a sequential solution to the
  density classification task, using an intermediate alphabet. 
\end{proof}

\section{Generalisations}

This solution is versatile enough to be generalised to support any finite
alphabet to compute a density over, and also to support configurations
in higher dimensions.
In this section we present the necessary changes to realise these
generalisations, along with proofs of their correctness.

\subsection{Larger alphabets}

Let us consider $S$ some finite alphabet. The generalisation of the density
classification task over that alphabet is the task of converging to the
uniform configuration $s^n$ if the symbol $s \in S$
is in strict majority over the
input configuration $x \in S^n$ of size $n$. Such a uniform configuration
must be a fixed point.

We denote $F_S$ the generalised solution, working on the extended
alphabet
$\Sigma = S \cup (\{ \circ, \bullet \} \times ( S \cup \{X\} ) \times 2^S)$,
where $2^S$ stands for the power set of $S$, the set of all the subsets of $S$.
This intermediate alphabet grows exponentially in the size of the input
alphabet $S$.
Similarly to the binary case, $F_S$ is a sequential cellular automaton
of radius $\frac{1}{2}$.

The rule cases are similar to the ones in the binary solution.
For any symbol $s \in S$, the rule $s\ s \mapsto s$ ensures that the uniform
configurations are fixed points. Any mismatch $s, s' \in S$
leads to the kickstart
rule case $s\ s' \mapsto \begin{psmallmatrix}\circ\\ X\\ \{s'\}\end{psmallmatrix}$. The propagation phases operate
as normal, moving symbols from the configuration into the active memory
when possible. The swap cases lead to another propagation phase
as long as the active memory contains more than one symbol, and to
a convergence phase if the memory contains exactly one symbol. An empty
active memory at the end of a propagation phase leads to undefined behaviour.

\begin{theorem}
  The sequential cellular automaton $F_S$ is a solution to the density classification task
  over the set $S$ using an intermediate alphabet.
\end{theorem}

\begin{proof}
  We follow the steps of the proof in the binary case, the main difference
  being that we count the number of propagation phases in a different way.
  In each propagation phase, the active memory accumulates one of each symbol
  of $S$ still in the configuration, and the propagation phases
  stop when the active memory is a singleton at the end of propagation.
  Such an event happens when each symbol that is not in majority in the initial
  configuration has been taken out. This time is exactly the maximum of the
  counts of the non-majority symbols, or the count of the second most
  present symbol in the initial configuration. After this exact amount
  of propagation phases, the configuration only contains $X$ and
  the majority symbols, and an extra propagation phase leads to a convergence phase. 
\end{proof}

\subsection{Higher dimensions}

For $d \geq 1$ an integer, we generalise the density classification task
over an alphabet $S$ to any number $d$ of dimensions. Instead of
considering cyclic configurations, we here consider cuboids of sides
$(n_1, n_2, \ldots, n_d)$
in $d$ dimensions, which associate a symbol in $S$ to any coordinate
$(x_1, x_2, \ldots, x_d)$. The density task in higher dimensions
then asks to replace all
symbols of this cuboids by the symbol which is in strict majority in the
initial configuration.

The generalisation of our solution on $d$-dimensional configurations,
denoted $F_{S, d}$, is still a sequential
rule, which explores the space of coordinates in an interative order.
We denote it $F_{S, d}$.
For example, for $d = 2$ and $n_1 = n_2 = 3$, the update sequence
is 
$((0, 0), (1, 0), (2, 0), (0, 1), (1, 1), (2, 1), (0, 2), (1, 2), (2, 2))$.

The neighbourhood of our solution is extended to contain one neighbour per
dimension, plus the cell itself.
For $(x_1, x_2, \ldots, x_d)$ the position of a given cell,
the neighbours of this cell are those with coordinates
$(x_1 - 1, x_2, \ldots, x_d)$, $(x_1 - 1, x_2 - 1, x_3, \ldots, x_d)$,
$(x_1 - 1, x_2 - 1, x_3 - 1, x_4, \ldots, x_d)$, $\ldots$,
$(x_1 - 1, x_2 - 1, \ldots, x_d - 1)$, plus $(x_1, x_2, \ldots, x_d)$, the
cell itself.
This neighbourhood is a subset of the Moore neighbourhood
in $d$ dimensions, as it always contains corners.

The local rule of $F_{S, d}$ is defined as the following
generalisation of the one-dimensional solution over $S$.
Firstly, the fixed point behaviour is applied when the neighbourhood is
of a uniform value in $S$.
Secondly, the kickstart rule case is applied
if the neighbourhood only contains symbols in $S$ but is not uniform.

Thirdly, the propagation phase must be applied by discerning which
neighbour was the cell to be last updated in the sequence;
or more precisely, it must determine what the active memory is.
This is done by considering all the neighbours which are in the intermediate
alphabet, and taking the larger memory, by comparing them by set inclusion.
We then apply the one-dimensional local rule as if this active memory
was in the cell to the left, and the present cell the cell to the right.

If both the present cells and all its neighbours are intermediate symbols
with a uniform internal counter, we apply the swap rule case. The active
memory is computed as before, by taking the maximum between the memories
of the neighbours.

Finally, if the present cell is in the intermediate
alphabet and has at
least one neighbour with a value in $S$,
we propagate that value to converge to a fixed point.

An execution of our generalised solution over a two dimensional configuration
of size $9$ over the alphabet $\{0, 1, 2\}$ is shown in
Figure~\ref{figure-exec2D}.

\begin{figure}[t!]
  \begin{center}
    \scalebox{0.8}{
\begin{tikzpicture}
  \newcommand{\xgap}{0.9}
  \newcommand{\ygap}{1.4}
  \newcommand{\largexgap}{3.1cm}
  \newcommand{\largeygap}{-4.8cm}

  \begin{scope}[xshift=\largexgap * 0, yshift=\largeygap * 0]
    \draw[gray!50!white] foreach \x in {0, ..., 3} {
      (\x * \xgap - \xgap / 2, - \ygap / 2) -- (\x * \xgap - \xgap / 2, \ygap * 2.5)
    };
    \draw[gray!50!white] foreach \y in {0, ..., 3} {
      (- \xgap / 2, \y * \ygap - \ygap / 2) -- (\xgap * 2.5, \y * \ygap - \ygap / 2)
    };

    \foreach \x\y\nodetext
    in {
      0/0/2, 1/0/2, 2/0/0,
      0/1/1, 1/1/2, 2/1/2,
      0/2/0, 1/2/1, 2/2/0
    }
    {
      \draw (\x * \xgap, \y * \ygap) node {\nodetext};
    }
  \end{scope}

  \begin{scope}[xshift=\largexgap * 1, yshift=\largeygap * 0]
    \draw[gray!50!white] foreach \x in {0, ..., 3} {
      (\x * \xgap - \xgap / 2, - \ygap / 2) -- (\x * \xgap - \xgap / 2, \ygap * 2.5)
    };
    \draw[gray!50!white] foreach \y in {0, ..., 3} {
      (- \xgap / 2, \y * \ygap - \ygap / 2) -- (\xgap * 2.5, \y * \ygap - \ygap / 2)
    };

    \draw (0 * \xgap, 2 * \ygap) node {0};

    \foreach \x\y\tapesymbol\memory
    in {
      0/0/2/S, 1/0/2/S, 2/0/0/S,
      1/1/X/S, 2/1/2/S,
      1/2/X/\{1\}
    }
    {
      \draw (\x * \xgap, \y * \ygap)
        node {$\begin{matrix}
          \circ \\
          \tapesymbol \\
          \memory
        \end{matrix}$};
    }

    \draw (0 * \xgap, 1 * \ygap)
      node {$\begin{matrix}
        \circ \\
        1 \\
        \{0, 1\}
      \end{matrix}$};

    \draw (2 * \xgap, 2 * \ygap)
      node {$\begin{matrix}
        \circ \\
        X \\
        \{0, 1\}
      \end{matrix}$};
  \end{scope}

  \begin{scope}[xshift=\largexgap * 2, yshift=\largeygap * 0]
    \draw[gray!50!white] foreach \x in {0, ..., 3} {
      (\x * \xgap - \xgap / 2, - \ygap / 2) -- (\x * \xgap - \xgap / 2, \ygap * 2.5)
    };
    \draw[gray!50!white] foreach \y in {0, ..., 3} {
      (- \xgap / 2, \y * \ygap - \ygap / 2) -- (\xgap * 2.5, \y * \ygap - \ygap / 2)
    };

    \draw (0 * \xgap, 2 * \ygap)
      node {$\begin{matrix}
        \circ \\
        0 \\
        S
      \end{matrix}$
    };

    \foreach \x\y\tapesymbol\memory
    in {
      2/0/X/S,
      0/1/X/\{1\}, 1/1/X/\{1\},
      1/2/X/\emptyset, 2/2/X/\emptyset
    }
    {
      \draw (\x * \xgap, \y * \ygap)
        node {$\begin{matrix}
          \bullet \\
          \tapesymbol \\
          \memory
        \end{matrix}$};
    }

    \draw (0 * \xgap, 0 * \ygap)
      node {$\begin{matrix}
        \bullet \\
        2 \\
        \{1, 2\}
      \end{matrix}$};

    \draw (1 * \xgap, 0 * \ygap)
      node {$\begin{matrix}
        \bullet \\
        2 \\
        \{1, 2\}
      \end{matrix}$};

    \draw (2 * \xgap, 1 * \ygap)
      node {$\begin{matrix}
        \bullet \\
        X \\
        \{1, 2\}
      \end{matrix}$};
  \end{scope}

  \begin{scope}[xshift=\largexgap * 3, yshift=\largeygap * 0]
    \draw[gray!50!white] foreach \x in {0, ..., 3} {
      (\x * \xgap - \xgap / 2, - \ygap / 2) -- (\x * \xgap - \xgap / 2, \ygap * 2.5)
    };
    \draw[gray!50!white] foreach \y in {0, ..., 3} {
      (- \xgap / 2, \y * \ygap - \ygap / 2) -- (\xgap * 2.5, \y * \ygap - \ygap / 2)
    };

    \draw (0 * \xgap, 2 * \ygap)
      node {$\begin{matrix}
        \bullet \\
        0 \\
        S
      \end{matrix}$
    };

    \foreach \x\y\tapesymbol\memory
    in {
      0/0/X/\{2\}, 1/0/2/\{2\}, 2/0/X/\{2\},
      0/1/X/\emptyset, 1/1/X/\emptyset, 2/1/X/\emptyset,
      1/2/X/\emptyset, 2/2/X/\emptyset
    }
    {
      \draw (\x * \xgap, \y * \ygap)
        node {$\begin{matrix}
          \circ \\
          \tapesymbol \\
          \memory
        \end{matrix}$};
    }
  \end{scope}

  \begin{scope}[xshift=\largexgap * 4, yshift=\largeygap * 0]
    \draw[gray!50!white] foreach \x in {0, ..., 3} {
      (\x * \xgap - \xgap / 2, - \ygap / 2) -- (\x * \xgap - \xgap / 2, \ygap * 2.5)
    };
    \draw[gray!50!white] foreach \y in {0, ..., 3} {
      (- \xgap / 2, \y * \ygap - \ygap / 2) -- (\xgap * 2.5, \y * \ygap - \ygap / 2)
    };

    \draw (0 * \xgap, 2 * \ygap)
      node {$\begin{matrix}
        \circ \\
        X \\
        \{0, 2\}
      \end{matrix}$
    };

    \foreach \x\y\tapesymbol\memory
    in {
      0/0/X/\emptyset, 1/0/X/\{2\}, 2/0/X/\{2\},
      0/1/X/\emptyset, 1/1/X/\emptyset, 2/1/X/\emptyset,
      1/2/X/\emptyset, 2/2/X/\emptyset
    }
    {
      \draw (\x * \xgap, \y * \ygap)
        node {$\begin{matrix}
          \bullet \\
          \tapesymbol \\
          \memory
        \end{matrix}$};
    }
  \end{scope}

  \begin{scope}[xshift=\largexgap * 0, yshift=\largeygap * 1]
    \draw[gray!50!white] foreach \x in {0, ..., 3} {
      (\x * \xgap - \xgap / 2, - \ygap / 2) -- (\x * \xgap - \xgap / 2, \ygap * 2.5)
    };
    \draw[gray!50!white] foreach \y in {0, ..., 3} {
      (- \xgap / 2, \y * \ygap - \ygap / 2) -- (\xgap * 2.5, \y * \ygap - \ygap / 2)
    };

    \draw (0 * \xgap, 2 * \ygap)
      node {$\begin{matrix}
        \bullet \\
        X \\
        \{2\}
      \end{matrix}$
    };

    \foreach \x\y\tapesymbol\memory
    in {
      0/0, 1/0, 2/0,
      0/1, 1/1, 2/1,
      1/2, 2/2
    }
    {
      \draw (\x * \xgap, \y * \ygap)
        node {$2$};
    }
  \end{scope}

  \begin{scope}[xshift=\largexgap * 1, yshift=\largeygap * 1]
    \draw[gray!50!white] foreach \x in {0, ..., 3} {
      (\x * \xgap - \xgap / 2, - \ygap / 2) -- (\x * \xgap - \xgap / 2, \ygap * 2.5)
    };
    \draw[gray!50!white] foreach \y in {0, ..., 3} {
      (- \xgap / 2, \y * \ygap - \ygap / 2) -- (\xgap * 2.5, \y * \ygap - \ygap / 2)
    };

    \foreach \x\y\tapesymbol\memory
    in {
      0/0, 1/0, 2/0,
      0/1, 1/1, 2/1,
      0/2, 1/2, 2/2
    }
    {
      \draw (\x * \xgap, \y * \ygap)
        node {$2$};
    }
  \end{scope}

  \begin{scope}[xshift=\largexgap * 4, yshift=\largeygap * 1]
    \draw[gray!50!white] foreach \x in {0, ..., 3} {
      (\x * \xgap - \xgap / 2, - \ygap / 2) -- (\x * \xgap - \xgap / 2, \ygap * 2.5)
    };
    \draw[gray!50!white] foreach \y in {0, ..., 3} {
      (- \xgap / 2, \y * \ygap - \ygap / 2) -- (\xgap * 2.5, \y * \ygap - \ygap / 2)
    };

    \foreach \x\y\ry
    in {
      0/0/2, 1/0/2, 2/0/2,
      0/1/1, 1/1/1, 2/1/1,
      0/2/0, 1/2/0, 2/2/0
    }
    {
      \draw (\x * \xgap, \y * \ygap)
        node {$(\x, \ry)$};
    }
  \end{scope}

\end{tikzpicture}
}
  \end{center}
  \caption{
    \label{figure-exec2D}
    An execution of $F_{\{0, 1, 2\}, 2}$ on a $3$ by $3$ configuration, with the set $S$ of the input alphabet being $\{0, 1, 2\}$.
    This execution lasts for $7$ configurations, which are illustrated
    from left to right, top to bottom. 
    Within each configuration, the sequential order is from left to right, and from top to bottom.
    The element at the bottom right
    of the figure indicates the coordinates of the cells; for instance,
    the neighbours of $(0, 0)$ are the cells $(2, 0)$, $(2, 2)$ and $(0, 0)$ itself.
    As all these cells have the value $0$ in the initial configuration,
    the cell $(0, 0)$ does not change value during the first update.
  }
\end{figure}

\begin{theorem}
  The sequential
  cellular automaton $F_{S, d}$ is a solution to the density classification
  task over the set $S$ in $d$ dimensions using an intermediate alphabet.
\end{theorem}

\begin{proof}
  This proof works by showing that $F_{S, d}$
  boils down to the application of $F_S$ to the $d$-dimensional configuration
  as if it were one-dimensional, by reconstructing the active memory from
  its neighborhood.
  We will show that this is the case on every phase of the solution.

  First, observe that if the configuration is uniform over $S$, then
  it is a fixed point of $F_{S, d}$.
  However, if the configuration is non-uniform over $S$,
  then there exist two locations that differ. We claim that if two
  locations differ, then a kickstart rule case happens in the first update
  of the configuration.
  This is true because a path can be made from any two locations of the
  configuration only using jumps in the local neighbourhood of the cells.
  If the starting and ending points of such a path have different values,
  then this path is bound to have two consecutive values that differ,
  and those two values would be the value of a cell and one of its neighbours,
  which leads to a kickstart rule case when this cell is updated.

  Let us assume that the intermediate symbol introduced in
  the kickstart rule case
  is in the cell
  of coordinates $(x_1, x_2, \ldots x_d)$.
  Let us show that a propagation phase now starts and replaces the entire
  configuration with intermediate symbols.

  If we assume that $x_1 < n_1 - 1$,
  then the next coordinate to be updated is $(x_1 + 1, x_2, \ldots, x_d)$.
  This new cell has not been updated yet since the kickstart, and thus
  contains a symbol in $S$. The first cell in its neighbourhood is
  the active cell, and thus the propagation rule case is applied.

  Let us now assume that $x_1 = n_1 - 1$, but that $x_2 < n_2 - 1$.
  This means that the next cell to
  be updated is of coordinates $(x_1 + 1, x_2 + 1, x_3, \ldots, x_d)$.
  Similarly to before, this cell has not yet been updated and contains a
  symbol over $S$. Moreover, its neighbour at coordinate
  $(x_1, x_2 + 1, x_3, \ldots, x_d)$ has not been updated either, and
  the neighbour at coordinate $(x_1, x_2, \ldots, x_d)$
  is the active cell, and contains an intermediate symbol.
  This means that we apply a propagation rule case; to do so, the
  local rule is defined to apply a propagation phase as in $F_S$ by
  taking the active memory as the largest memory within its neighbourhood.
  As the memory only accumulates symbols
  throughout a propagation phase, taking the largest set (in the sense
  of set inclusion) is always a defined operation that computes
  the correct active memory.
  Note that the configuration symbol of the active cell (in $S \cup X$)
  is actually never relevant to apply the local rule, and it is not necessary
  to know its value. It follows that one step of propagation is applied
  as it would in the one-dimensional solution $F_S$.

  This logic works for any number of dimensions; when the sequential order
  jumps in the direction of the $k$-th dimension, the $k - 1$ first neighbours
  of the new cell have not yet been updated in the current propagation phase,
  but the $k$-th neighbour is the
  active cell, the last cell to be updated.
  This means that the propagation phase operates as the one
  dimensional local rule would, up to the point where the entire configuration
  has been converted to the intermediate alphabet with a uniform internal counter.
  The new cell to be updated has the active cell within its neighbourhood,
  which means that it can accurately compute the active memory by taking
  the largest memory in its neighbourhood. If this active
  memory is a singleton, this means that the previous propagation phases
  removed all symbols except one, and the solution can converge. If the active memory
  is empty, then the density is not defined over the initial configuration and
  the behaviour is said to be undefined.
  If the active memory has size at least $2$,
  then the solution starts a new propagation phase by emptying the memory
  and updating the internal counter.

  Overall, this solution successfully explores a $d$-dimensional configuration
  to apply the local rule as if it were one-dimensional, therefore leading to the result. 
\end{proof}

\section{Closing words}

The presented solution pushes the boundary of the classification tasks known
to be possible using cellular automata. It proves that using an
intermediate alphabet together with a sequential update schedule is
enough to solve the density classification task.

We are of the opinion that a solution of this exact problem using the
synchronous update schedule does not exist. Proving this conjecture would
be a fitting continuation of the present work, and would confirm that the
present solution is in fact one example of some minimal requirements
that are sufficient to solve the density classification task.

\section*{Acknowledgements}

P.P.B. thanks the Brazilian agencies CNPq (Conselho Nacional de Desenvolvimento Cient\'{i}fico e Tecnol\'{o}gico) for the research grant PQ 303356/2022-7, and CAPES (Coordena\c{c}\~{a}o de Aperfei\c{c}oamento de Pessoal de N\'{i}vel Superior) for Mackenzie-PrInt research grant no.\ 88887.310281/2018-00. P.P.B. and E.R. jointly thank CAPES for the research grant STIC-AmSud no.\ 88881.694458/2022-01, and P.P. thanks CAPES for the postdoc grant no.\ 88887.833212/2023-00.

\bibliographystyle{plain}
{\small{\bibliography{bib}}}

\begin{thebibliography}{10}

\bibitem{CA2}
M.~M. Aburas, Y.~M. Ho, M.~F. Ramli, and Z.~H. Ash’aari.
\newblock The simulation and prediction of spatio-temporal urban growth trends
  using cellular automata models: A review.
\newblock {\em International Journal of Applied Earth Observation and
  Geoinformation}, 52:380--389, 2016.

\bibitem{synchParitySolutionWith150}
P.~P. Balbi, E.~Ruivo, and F.~Faria.
\newblock Synchronous solution of the parity problem on cyclic configurations,
  with elementary cellular automaton rule 150, over a family of directed,
  non-circulant, regular graphs.
\newblock {\em Information Sciences}, 615:578--603, 2022.

\bibitem{betel2013}
H.~Betel, P.~P.~B. de~Oliveira, and P.~Flocchini.
\newblock Solving the parity problem in one-dimensional cellular automata.
\newblock {\em Natural Computing}, 12:323--337, 2013.

\bibitem{Briceno2013}
R.~Briceno, P.M. de~Espan{\'e}s, A.~Osses, and I.~Rapaport.
\newblock Solving the density classification problem with a large diffusion and
  small amplification cellular automaton.
\newblock {\em Physica D: Nonlinear Phenomena}, 261:70--80, 2013.

\bibitem{gen_SynchParitySolution_on_Graph}
Fernando Faria, Eurico Ruivo, and Pedro~Paulo Balbi.
\newblock Generalisation of a synchronous solution of the parity problem on
  cyclic configurations over a non-circulant graph.
\newblock {\em Information Sciences}, 666:120387, 2024.

\bibitem{Fates2011}
N.A. Fat{\`e}s.
\newblock Stochastic cellular automata solve the density classification problem
  with an arbitrary precision.
\newblock In {\em Symposium on Theoretical Aspects of Computer
  Science-STACS2011}, volume~9, pages 284--295. Schloss
  Dagstuhl--Leibniz-Zentrum fuer Informatik, 2011.

\bibitem{CA31}
A.~G. Hoekstra, J.~Kroc, and P.~M.~A. Sloot.
\newblock Introduction to modeling of complex systems using cellular automata.
\newblock In {\em Simulating Complex Systems by Cellular Automata}, pages
  1--16. Springer, 2010.

\bibitem{kari2012}
J.~Kari and B.~Le~Gloannec.
\newblock Modified traffic cellular automaton for the density classification
  task.
\newblock {\em Fundamenta Informaticae}, 116(1-4):141--156, 2012.

\bibitem{Kutrib2014}
M.~Kutrib.
\newblock Complexity of one-way cellular automata.
\newblock In {\em Cellular Automata and Discrete Complex Systems}, pages 3--18,
  Cham, 2015. Springer, Springer International Publishing.

\bibitem{land1995}
M.~Land and R.~K. Belew.
\newblock No perfect two-state cellular automata for density classification
  exists.
\newblock {\em Physical review letters}, 74(25):5148, 1995.

\bibitem{Leal2023}
L.~Leal, P.~Montealegre, A.~Osses, and I.~Rapaport.
\newblock A large diffusion and small amplification dynamics for density
  classification on graphs.
\newblock {\em International Journal of Modern Physics C}, 34(05):2350056,
  2023.

\bibitem{LeeXuChau}
K.M. Lee, H.~Xu, and H.F. Chau.
\newblock Parity problem with a cellular automaton solution.
\newblock {\em Physical Review E}, 64(2):026702, 2001.

\bibitem{martins2005}
C.L.M. Martins and P.~P.~B. de~Oliveira.
\newblock Evolving sequential combinations of elementary cellular automata
  rules.
\newblock In {\em European Conference on Artificial Life}, pages 461--470.
  Springer, 2005.

\bibitem{CA39}
A.~R. Mikler, S.~Venkatachalam, and K.~Abbas.
\newblock Modeling infectious diseases using global stochastic cellular
  automata.
\newblock {\em Journal of Biological Systems}, 13(04):421--439, 2005.

\bibitem{balbi11}
M.~Montalva-Medel, P.~P.~B. de~Oliveira, and E.~Goles.
\newblock A portfolio of classification problems by one-dimensional cellular
  automata, over cyclic binary configurations and parallel update.
\newblock {\em Natural Computing}, 17:663--671, 2018.

\bibitem{Ninagawa}
S.~Ninagawa.
\newblock Solving the parity problem with elementary cellular automaton rule
  60.
\newblock {\em Proc. of AUTOMATA and JAC}, 2012, 2012.

\bibitem{densityChecker}
P.~Perrotin.
\newblock {Density solution checker}.
\newblock \url{https://github.com/Demonagon/density_checker}, August 2024.

\bibitem{balbi14}
E.~L.~P. Ruivo and P.~P.~B. de~Oliveira.
\newblock A perfect solution to the parity problem with elementary cellular
  automaton 150 under asynchronous update.
\newblock {\em Information Sciences}, 493:138--151, 2019.

\bibitem{CA54}
D.~A. Wolf-Gladrow.
\newblock {\em Lattice-gas cellular automata and lattice Boltzmann models: an
  introduction}.
\newblock Springer, 2004.

\end{thebibliography}

\end{document}